\documentclass[11pt,reqno]{amsart}
\usepackage{fullpage}
\usepackage[foot]{amsaddr}
\usepackage[OT1]{fontenc}
\usepackage{amsmath}
\usepackage{amssymb}
\usepackage{amsthm}
\usepackage{thmtools} 
\usepackage{algorithm}
\usepackage{algpseudocode}
\usepackage{enumitem}
\usepackage{xifthen}
\usepackage{xspace}

\usepackage[margin=1cm]{caption} 
\usepackage{subfig}

\usepackage[thinlines]{easytable}

\usepackage[bookmarks=true,hypertexnames=false,pagebackref]{hyperref}
\hypersetup{colorlinks=true, citecolor=blue, linkcolor=red, urlcolor=blue}

\usepackage{pgfplots}
\pgfplotsset{compat=1.16}
\usepackage{tikz}
\usetikzlibrary{arrows,arrows.meta,backgrounds,calc,fit,decorations.pathreplacing,decorations.markings,shapes.geometric}

\tikzstyle{internal} = [draw, fill, shape=circle]
\tikzstyle{external} = [shape=circle]
\tikzstyle{square}   = [draw, fill, rectangle]
\tikzstyle{triangle} = [draw, fill, regular polygon, regular polygon sides=3, inner sep=3pt]
\tikzstyle{pentagon} = [draw, fill, regular polygon, regular polygon sides=5, inner sep=2pt, minimum size=14pt]
\tikzset{every fit/.append style=text badly centered}

\usetikzlibrary{positioning,chains,fit,shapes,calc}
\usetikzlibrary{trees}
\usetikzlibrary{decorations.pathreplacing}
\usetikzlibrary{decorations.pathmorphing}
\usetikzlibrary{decorations.markings}
\tikzset{>=latex} 

\usepackage{ifthen}

\usepackage{cleveref}

\usepackage[textsize=tiny]{todonotes}

\usepackage[normalem]{ulem}

\usepackage{mleftright}

\usepackage{cool}
\Style{DSymb={\mathrm d},DShorten=true,IntegrateDifferentialDSymb=\mathrm{d}}

\newcommand{\Ex}{\mathop{\mathbb{{}E}}\nolimits}
\renewcommand{\Pr}{\mathop{\mathrm{Pr}}\nolimits}

\def\*#1{\mathbf{#1}}
\def\+#1{\mathcal{#1}}
\def\-#1{\mathrm{#1}}
\def\=#1{\mathbb{#1}}
\def\^#1{\mathbb{#1}}

\newcommand{\abs}[1]{\ensuremath{\left\vert#1\right\vert}}

\newcommand{\eps}{\varepsilon}

\newcommand{\Var}[2]{\ensuremath{\textnormal{Var}_{#1}\left(#2\right)}}

\newcommand{\defeq}{:=}

\newcommand{\numP}{\#{\textnormal{\textbf{P}}}}

\newtheorem{theorem}{Theorem}

\newtheorem{lemma}[theorem]{Lemma}

\theoremstyle{definition}

\theoremstyle{remark}

\crefname{theorem}{Theorem}{Theorems}
\crefname{observation}{Observation}{Observations}
\crefname{claim}{Claim}{Claims}
\crefname{condition}{Condition}{Conditions}
\crefname{algorithm}{Algorithm}{Algorithms}
\crefname{property}{Property}{Properties}
\crefname{example}{Example}{Examples}
\crefname{fact}{Fact}{Facts}
\crefname{lemma}{Lemma}{Lemmas}
\crefname{corollary}{Corollary}{Corollaries}
\crefname{definition}{Definition}{Definitions}
\crefname{remark}{Remark}{Remarks}
\crefname{proposition}{Proposition}{Propositions}
\crefname{equation}{equation}{equations}
\crefname{enumi}{Case}{Case}
\creflabelformat{enumi}{(#2#1#3)}

\makeatletter
\def\prob#1#2#3{\goodbreak\begin{list}{}{\labelwidth\z@ \itemindent-\leftmargin
      \itemsep\z@  \topsep6\p@\@plus6\p@
      \let\makelabel\descriptionlabel}
  \item[\textbf{Name}]#1
  \item[\textbf{Instance}]#2
  \item[\textbf{Output}]#3
  \end{list}}
\makeatother

\makeatletter
\providecommand\@dotsep{5}
\def\listtodoname{Todo list}
\def\listoftodos{\@starttoc{tdo}\listtodoname}
\makeatother

\newcommand{\dTV}{d_{\mathrm{TV}}}

\definecolor{ABcolor}{RGB}{0,0,255}

\definecolor{HGcolor}{RGB}{255,50,50}

\definecolor{KAcolor}{RGB}{34,139,34}

\title{Linear time approximation of the TV distance between product distributions}

\begin{document}

\author{Konrad Anand}
\author{Alistair Benford}
\author{Heng Guo}
\address{School of Informatics, University of Edinburgh, Informatics Forum, Edinburgh, EH8 9AB, United Kingdom}

\begin{abstract}
  We present a linear time approximation algorithm of the total variation distance between two product distributions.
  The main algorithm was found using ChatGPT 5.6 Sol Ultra.
\end{abstract}

\maketitle

\section{Introduction}

The total variation (TV) distance is an important statistic with wide-ranging applications.
However, the computation of TV distances is not easy, even for simple distributions such as product distributions.
Unlike the Kullback-Leibler divergence (also known as relative entropy), the TV distance does not tensorise over product distributions.
In fact, the exact evaluation for product distributions turns out to be \numP-hard \cite{BGMMPV25b}.

It is then natural to ask for approximation.
Chernoff bounds imply relatively straightforward efficient additive error approximations.
(See, e.g.\ \cite{PTXZ26} and the references therein for additively approximate TV distances in more advanced models.)
Instead, we focus on relative errors.
In this context, the first efficient (randomised) approximation algorithm was given by Feng, Guo, Jerrum, and Wang \cite{FGJW23},
which has inspired a number of follow-up works in various settings \cite{BGMMPV24, BGMMPV25a, BFS25, FLY25, FF26, FFYZ26}.
Most notably, Feng, Liu, and Liu \cite{FLL24} derandomised the algorithm in \cite{FGJW23}.

In this note, we give a linear time approximation algorithm of the total variation distance between two product distributions,
improving upon \cite{FGJW23}.

\begin{theorem}\label{thm:main}
  Let $q,n>0$ be two integers, and $ P=\bigotimes_{i=1}^{n}P_i$ and $Q=\bigotimes_{i=1}^{n}Q_i$ be two product distributions on $[q]^n$, given explicitly by the marginal probabilities of each coordinate. 
  For any $0<\varepsilon\leq 1$ and $0<\delta<1$, there is a randomised algorithm that outputs a nonnegative random variable $Z$ satisfying
  \[
    \Pr\left[
      (1-\varepsilon)d_{\mathrm{TV}}(P,Q)
      \leq Z\leq
      (1+\varepsilon)d_{\mathrm{TV}}(P,Q)
    \right]\geq 1-\delta
  \]
  in $O\left(qn\,\varepsilon^{-2}\log\frac{1}{\delta}\right)$ time.
\end{theorem}
For \Cref{thm:main}, real arithmetic operations and generating a uniform real from $(0,1)$ are considered unit-cost, and $P$ and $Q$ are given by the marginals $p_{ij}$ and $q_{ij}$ for $i\in[n]$ and $j\in[q]$.
Thus, for fixed $\eps$ and $\delta$, the running time of \Cref{thm:main} is linear in the input size.
Previously, the algorithm in \cite{FGJW23} states an $O\left( \frac{q n^2}{\eps^2}\log\frac{1}{\delta} \right)$ time,
although it can be improved to $O\left( \frac{q n^{1.5}}{\eps^2}\log\frac{1}{\delta} \right)$ time using a bound given by Kontorovich \cite{Kon25}.
The deterministic algorithm of Feng, Liu, and Liu runs in time $O\left( \frac{q n^{2}}{\eps}\log q\log\frac{n}{\eps \dTV(P,Q)} \right)$.
Since the ideas of \cite{FGJW23} have been found useful in various other settings,
we hope the ideas of \Cref{thm:main} will find wide applications too.

The algorithm in \cite{FGJW23} relies on an estimator based on couplings,
and the algorithm in \cite{FLL24} is based on the idea of careful discretisation.
Instead, the algorithm in \Cref{thm:main} uses very different ideas.
Essentially, the improved running time comes from an adaptation of the filtered Monte Carlo technique \cite{Gla93}.

In \Cref{sec:limit}, we present an alternative implementation of \Cref{thm:main}, which runs in time $O\left( nq\log q+\frac{n\log q}{\eps^2}\log\frac{1}{\delta} \right)$.
This is faster if $\eps^{-2}\log q+\frac{\log q}{q}=o(1)$.
On the other hand, we also present an $\Omega(qn)$ lower bound in \Cref{thm:lb}.

In fact, the authors of this note first set out to improve the run-time bounds in \cite{FGJW23}, and found a roughly $O(qn^{1.3})$ time algorithm for the product of Bernoulli distributions.
Shortly after that, the authors tried ChatGPT 5.6 Sol Ultra, which managed to find the linear-time (and much more elegant) algorithm in \Cref{thm:main} with a single prompt.
The prompt is adapted from OpenAI's cycle double cover prompt by only changing the problem statement.
However, the key reference \cite{Gla93} did not emerge until after a few more rounds of interaction with ChatGPT.
In the note we still use words like ``we'' or ``our'', mostly out of habit, although the main ideas actually come from ChatGPT.
The proofs are written or rewritten by the authors, who assume all responsibility for errors.

\section{The algorithm}

Let $\Omega$ be the state space, which in our setting is $[q]^n$.
We may assume $\dTV(P,Q)\neq 0$, as it is easy to rule out otherwise at the start of the algorithm.
The starting point of the algorithm is to rewrite 
\begin{align}
  \dTV(P,Q) & = \frac{1}{2}\sum_{\omega\in\Omega}\abs{P(\omega)-Q(\omega)}\notag\\
  &=\sum_{\omega\in\Omega}\frac{P(\omega)+Q(\omega)}{2} \frac{|P(\omega)-Q(\omega)|}{P(\omega)+Q(\omega)}. \label{eqn:U}
\end{align}
Therefore, if we sample $\omega$ from the mixture $\frac{P+Q}{2}$, then $\frac{|P(\omega)-Q(\omega)|}{P(\omega)+Q(\omega)}$ is an unbiased estimator.
However, this estimator has a large relative variance, especially when $\dTV(P,Q)$ is small.
The main point of our algorithm is to overcome this issue by designing a different estimator.

We use a filtered Monte Carlo idea \cite{Gla93}.
The algorithm first draws a hidden fair hypothesis $H\in\{+,-\}$. 
Then, draw a sample from $P$ if $H=+$ and from $Q$ otherwise.
Equivalently, this is drawing from $\frac{P+Q}{2}$.
Call the samples $X=(X_1,\dots,X_n)$.
After each observed coordinate, maintain the posterior probabilities
\[
\alpha_i=\Pr(H=+\mid \+F_i)
\qquad\text{and}\qquad
\beta_i=1-\alpha_i,
\]
where the filtration $\+F_i$ is the $\sigma$-algebra generated by $X_1,\dots,X_i$.
It is important that $\+F_i$ excludes the randomness of $H$.
If the sample up to coordinate $i$ is $\omega_{\le i}$,
then 
\begin{align*}
  \alpha_i=\frac{P(\omega_{\le i})}{P(\omega_{\le i})+Q(\omega_{\le i})}.
\end{align*}
The posterior bias $U_i=\alpha_i-\beta_i$ is a bounded martingale. 
Its terminal absolute value, by \eqref{eqn:U}, has expectation exactly equal to the total variation distance:
\[
\Ex |U_n|=d_{\mathrm{TV}}(P,Q).
\]
Instead of using $|U_n|$ itself, whose relative variance can be large when the distance is small,
the algorithm sums the predictable increases of the submartingale $|U_i|$:
\[
  C_i\defeq\Ex \left[|U_i|-|U_{i-1}|\mid\mathcal{F}_{i-1}\right] = \mathbb{E}\!\left[|U_i|\mid\mathcal{F}_{i-1}\right]-|U_{i-1}|.
\]
The final estimator is $A\defeq\sum_{i=1}^{n}C_i$.
The way we compute $A$ is given by \Cref{alg:main-alg},
where we denote $(x)_+\defeq\max\{x,0\}$. 

\begin{algorithm}[htbp]
\caption{The estimator}\label{alg:main-alg}
\begin{algorithmic}[1]
  \State Draw uniformly at random a sign $H\in\{+,-\}$. \Comment{The same sign is used for all coordinates.}
\State Initialize $\alpha\gets \tfrac12$, $\beta\gets \tfrac12$, and $A\gets 0$.
\For{$i=1,\ldots,n$}
    \State Before observing coordinate $i$, compute
    \Statex
    \begin{align}      \label{eqn:C_i}
          C_i=
          \begin{cases}
            \displaystyle 2\sum_{c=1}^{q}\bigl(\beta Q_i(c)-\alpha P_i(c)\bigr)_+,
            & \text{ if }\alpha\geq\beta,\\
            \displaystyle 2\sum_{c=1}^{q}\bigl(\alpha P_i(c)-\beta Q_i(c)\bigr)_+,
            & \text{ if } \alpha<\beta.
          \end{cases}
    \end{align}
    \State $A\gets A+C_i$.
    \If{$H=+$}
        \State Sample $X_i$ from $P_i$.
    \Else
        \State Sample $X_i$ from $Q_i$.
    \EndIf
    \State Update simultaneously \label{line:update}
    \Statex
    \[
    \alpha\gets\frac{\alpha P_i(X_i)}{\alpha P_i(X_i)+\beta Q_i(X_i)}
    \qquad \text{and}
    \qquad
    \beta\gets\frac{\beta Q_i(X_i)}{\alpha P_i(X_i)+\beta Q_i(X_i)}.
    \]
\EndFor
\State \Return $A$.
\end{algorithmic}
\end{algorithm}

The estimator $A$ is nonnegative and unbiased. 
\begin{lemma}  \label{lem:first-moment}
  Let $A$ be as in \Cref{alg:main-alg}.
  Then, $\Ex [A] = \dTV(P,Q)$.
\end{lemma}

The key property of $A$ is the following second moment bound.
\begin{lemma}  \label{lem:second-moment}
  Let $A$ be as in \Cref{alg:main-alg}.
  Then, $\Ex A^2\le 2\dTV(P,Q)^2$,
  which implies that the relative variance $\frac{\Var{}{A}}{(\Ex A)^2} \le 1$.
\end{lemma}

Given \Cref{lem:first-moment} and \Cref{lem:second-moment}, we can prove \Cref{thm:main}.

\begin{proof}[Proof of \Cref{thm:main}]
  Let $s\defeq\left\lceil\frac{4}{\varepsilon^2}\right\rceil$ and $r$ be the smallest odd integer with $r\geq 8\log\frac{1}{\delta}$.
  Run $s$ independent estimators using \Cref{alg:main-alg}, and take the mean.
  Repeat this process $r$ times and output the median of all $r$ means.
  This is the standard median of means estimator.
  The approximation guarantee follows from \Cref{lem:first-moment}, \Cref{lem:second-moment}, Chebyshev's inequality, and the Chernoff bound.

  As for the running time, it is not hard to see that the time to generate one estimator $A$ is $O(nq)$,
  and the theorem holds.
\end{proof}

\section{The analysis}

The analysis relies on the following lemma, which holds for arbitrary distributions, not only products.

\begin{lemma}\label{lem:prior}
Let $R$ and $S$ be probability distributions on a finite set $\Omega$. Let $\alpha,\beta\geq 0$ satisfy
$\alpha+\beta=1$, and $u\defeq\alpha-\beta$. Define
\[
  V_{R,S}(u)\defeq\sum_{\omega\in\Omega}|\alpha R(\omega)-\beta S(\omega)|-|u|.
\]
Then
\begin{equation}\label{eq:prior-1}
  V_{R,S}(u)=\begin{cases}
            \displaystyle 2\sum_{\omega\in\Omega}\bigl(\beta S(\omega)-\alpha R(\omega)\bigr)_+,
            & \text{ if }\alpha\geq\beta,\\
            \displaystyle 2\sum_{\omega\in\Omega}\bigl(\alpha R(\omega)-\beta S(\omega)\bigr)_+,
            & \text{ if } \alpha<\beta,
          \end{cases}
\end{equation}
and
\begin{equation}\label{eq:prior-2}
    0\leq V_{R,S}(u)\leq (1-|u|)d_{\mathrm{TV}}(R,S).
\end{equation}

\end{lemma}

\begin{proof}
We first assume $u\geq 0$, so $\alpha\geq\beta$. Note that for any numbers $a_1,\ldots,a_m$,
\begin{equation*}
  \sum_{i=1}^m|a_i|-\sum_{i=1}^ma_i=\sum_{i=1}^m((a_i)_+ + (-a_i)_+)-\sum_{i=1}^m((a_i)_+ - (-a_i)_+)=2\sum_{i=1}^m(-a_i)_+.
\end{equation*}
Therefore, it holds that
\begin{align*}
  V_{R,S}(u)&=\sum_{\omega\in\Omega}|\alpha R(\omega)-\beta S(\omega)|-(\alpha-\beta)=\sum_{\omega\in\Omega}|\alpha R(\omega)-\beta S(\omega)|-\sum_{\omega\in\Omega}(\alpha R(\omega)-\beta S(\omega))\\
  &=2\sum_{\omega\in\Omega}\bigl(\beta S(\omega)-\alpha R(\omega)\bigr)_+,
\end{align*}
which implies \eqref{eq:prior-1} and the lower bound of \eqref{eq:prior-2}.
Moreover,
\[
\beta S(\omega)-\alpha R(\omega)
=\beta\bigl(S(\omega)-R(\omega)\bigr)-(\alpha-\beta)R(\omega)
\leq\beta\bigl(S(\omega)-R(\omega)\bigr).
\]
It follows that
\begin{align*}
V_{R,S}(u)
&\leq 2\beta\sum_{\omega\in\Omega}\bigl(S(\omega)-R(\omega)\bigr)_+\\
&=2\beta d_{\mathrm{TV}}(R,S)=(1-u)d_{\mathrm{TV}}(R,S).
\end{align*}
The case $u\leq 0$ is symmetric after exchanging $R$ and $S$.
The lemma follows.
\end{proof}

We now analyse the first moment and show \Cref{lem:first-moment}.
\begin{proof}[Proof of \Cref{lem:first-moment}]
  Recall the definitions of $\alpha_i$, $\beta_i$, and $U_i$.
  In \Cref{alg:main-alg}, at coordinate $i$, 
  the probability of drawing $c\in[q]$ is $\alpha_{i-1} P_i(c)+\beta_{i-1}Q_i(c)$.
  Thus, by Bayes's rule,
  \begin{align*}
    \alpha_i&=\frac{\alpha_{i-1} P_i(c)}{\alpha_{i-1} P_i(c)+\beta_{i-1}Q_i(c)},
    &\beta_i&=\frac{\beta_{i-1} Q_i(c)}{\alpha_{i-1} P_i(c)+\beta_{i-1}Q_i(c)}.
  \end{align*}
  Note that this is consistent with the update rule in Line \ref{line:update}.
  Thus, inductively, the values $\alpha$ and $\beta$ at the end of iteration $i$ are exactly $\alpha_i$ and $\beta_i$.

  Since $\Ex\left[\abs{U_i}\mid\mathcal{F}_{i-1}\right] =\sum_{c=1}^{q}\abs{\alpha_{i-1} P_i(c)-\beta_{i-1} Q_i(c)}$, identifying $R=P_i$ and $S=Q_i$ in \eqref{eq:prior-1} yields $C_i=\Ex\left[\abs{U_i}\mid\mathcal{F}_{i-1}\right]-\abs{U_{i-1}}$,
  where $C_i$ is the value in \Cref{alg:main-alg}.
  Taking expectations and telescoping gives $\Ex A =\Ex\abs{U_n}$.
  For every $\omega$ so that $P(\omega)+Q(\omega)>0$,
  \[
    \abs{U_n(\omega)}=\frac{\abs{P(\omega)-Q(\omega)}}{P(\omega)+Q(\omega)}.
  \]
  The probability of drawing $\omega$ is $\frac{P(\omega)+Q(\omega)}{2}$.
  Therefore,
  \begin{align*}
    \Ex A &=\sum_{\omega\in\Omega}\frac{P(\omega)+Q(\omega)}{2} \frac{\abs{P(\omega)-Q(\omega)}}{P(\omega)+Q(\omega)}=\dTV(P,Q). \qedhere
  \end{align*}
\end{proof}

Finally, we complete the proof by analysing the second moment.

\begin{proof}[Proof of \Cref{lem:second-moment}]
Let $P_{i:n}=\bigotimes_{j=i}^{n}P_j$ and $Q_{i:n}=\bigotimes_{j=i}^{n}Q_j$.
Since $P$ and $Q$ are product distributions,
conditioned on $\mathcal{F}_{i-1}$, 
the remaining experiment consists of distinguishing $P_{i:n}$ versus $Q_{i:n}$ with current prior bias $U_{i-1}$. 
By telescoping conditional expectations and applying \eqref{eq:prior-2},
\begin{equation}
\Ex\left[\left.\sum_{j=i}^{n}C_j\,\right|\,\mathcal{F}_{i-1}\right]
=V_{P_{i:n},Q_{i:n}}(U_{i-1})
\leq \dTV(P_{i:n},Q_{i:n})
\leq \dTV(P,Q),
\label{eq:tail-compensator}
\end{equation}
where the last inequality follows because projecting on the $i,\dots,n$ coordinates cannot increase the TV distance. 

Next we rewrite
\[
A^2
=2\sum_{i=1}^{n}C_i\sum_{j=i}^{n}C_j
 -\sum_{i=1}^{n}C_i^2.
\]
The key point is that $C_i$ is completely determined by $\+F_{i-1}$.
Thus,
\begin{align}
  \Ex\left[\left.C_i\sum_{j=i}^{n}C_j\,\right|\,\mathcal{F}_{i-1}\right] = \Ex\left[\left.C_i\,\right|\,\mathcal{F}_{i-1}\right] \Ex\left[\left.\sum_{j=i}^{n}C_j\,\right|\,\mathcal{F}_{i-1}\right].
  \label{eqn:conditional}
\end{align}
Therefore,
\begin{align*}
  \Ex A^2 & = 2\sum_{i=1}^{n}\Ex\left[C_i\sum_{j=i}^{n}C_j\right]-\sum_{i=1}^{n}\Ex C_i^2\\
  & = 2\sum_{i=1}^{n}\Ex\left[\Ex\left[\left.C_i\sum_{j=i}^{n}C_j\,\right|\,\mathcal{F}_{i-1}\right]\right]-\sum_{i=1}^{n}\Ex C_i^2\\
  & = 2\sum_{i=1}^{n}\Ex\left[\Ex\left[\left.C_i\,\right|\,\mathcal{F}_{i-1}\right] \Ex\left[\left.\sum_{j=i}^{n}C_j\,\right|\,\mathcal{F}_{i-1}\right]\right]-\sum_{i=1}^{n}\Ex C_i^2 \tag{by \eqref{eqn:conditional}}\\
  & \le 2\dTV(P,Q)\sum_{i=1}^{n}\Ex\left[\Ex\left[\left.C_i\,\right|\,\mathcal{F}_{i-1}\right]\right]-\sum_{i=1}^{n}\Ex C_i^2 \tag{by \eqref{eq:tail-compensator}}\\
  & \le 2\dTV(P,Q)\sum_{i=1}^{n} \Ex C_i = 2\dTV(P,Q) \Ex A \\
  & = 2\dTV(P,Q)^2. \tag{by \Cref{lem:first-moment}}
\end{align*}
The lemma follows.
\end{proof}

\section{Further improved run-time and limits}
\label{sec:limit}

We first present an alternative way of implementing \Cref{thm:main}, which is faster if $\eps^{-2}\log q+\frac{\log q}{q}=o(1)$.
The key point is the computation of \eqref{eqn:C_i} mainly cares about what $c\in [q]$ makes $\frac{P_i(c)}{Q_i(c)}\ge \frac{\beta}{\alpha}$.
We may thus first sort for each $i$, the ratio $\frac{P_i(c)}{Q_i(c)}$.
Suppose the resulting ordering for coordinate $i$ is $x_{i,1},\dots,x_{i,q}$.
Then compute the suffix sum $\sum_{j\ge t} P_i(x_{i,j})$ for every $i\in[n]$ and $t\in[q]$, and similarly for the prefix sum, and for $Q_i$'s.
Then, at each iteration $i$, we use a binary search to locate the threshold, and then use the corresponding prefix or suffix sums to compute $C_i$.
The preprocessing per coordinate takes time $O(q\log q)$, and the binary search takes time $O(\log q)$.
The overall runtime is $O\left( nq\log q+\frac{n\log q}{\eps^2}\log\frac{1}{\delta} \right)$.

On the other hand, it is straightforward that any efficient approximation needs to read a constant fraction of the input at least.
Here we include a proof for completeness.
Our model is that $P$ and $Q$ are given by the $2qn$ marginals, and that the algorithm queries these marginals.

\begin{theorem}
\label{thm:lb}
Let $n\geq 1$, $q\geq 2$, $0<\varepsilon<1$, and
$0\leq\delta<1/2$.
In the marginal-query model, every randomised algorithm that outputs
a multiplicative $(1\pm\varepsilon)$-approximation to
$\dTV(P,Q)$ with probability at least $1-\delta$ for any two product distributions $P$ and $Q$ on $[q]^n$ must make
at least
\[
  (1-2\delta)n\left\lfloor\frac q2\right\rfloor
\]
queries in the worst case.
In particular, its running time is $\Omega(nq)$.
\end{theorem}

\begin{proof}
Let $m\defeq\left\lfloor\frac q2\right\rfloor$, $N\defeq nm$, and $\eta\defeq\frac{1}{2q}$.
Let $U$ be the uniform distribution on $[q]$, and let $ P=Q^{(0)}\defeq U^{\otimes n}$.
For each $(i,j)\in[n]\times[m]$, define $Q^{(i,j)}$ by taking all
marginals other than the $i$-th one to be uniform and setting
\[
  Q^{(i,j)}_i(2j-1)=\frac1q+\eta,
  \qquad
  Q^{(i,j)}_i(2j)=\frac1q-\eta,
\]
with all other entries of $Q^{(i,j)}_i$ equal to $1/q$.
For any $(i,j)\in[n]\times[m]$, only one marginal of $Q^{(i,j)}$ differs from $P$, and hence
\[
  \dTV(P,Q^{(i,j)})
  =\dTV(U,Q^{(i,j)}_i)
  =\eta.
\]
In contrast, $\dTV(P,Q^{(0)})=0$.

Fix the internal randomness of the algorithm, and consider its
execution on $Q^{(0)}$. Suppose this execution
makes at most $T$ queries. Let $S\subseteq[n]\times[m]$ contain
$(i,j)$ whenever the algorithm queries at least one of the cells
$Q_i(2j-1)$ and $Q_i(2j)$. Then $|S|\leq T$.
For every $(i,j)\notin S$, the execution log on $Q^{(i,j)}$ is identical to the log on $Q^{(0)}$.
It means that the algorithm returns the same value on these two inputs.

Consider a random input that is $Q^{(0)}$ with probability $1/2$
and is a uniformly random $Q^{(i,j)}$ with probability $1/2$.
If the algorithm outputs zero on $Q^{(0)}$, then it fails
on every $Q^{(i,j)}$ with $(i,j)\notin S$. 
Its average success probability is therefore at most
\[
  \frac12+\frac{|S|}{2N}
  \leq \frac12+\frac{T}{2N}.
\]
If it outputs a nonzero value on $Q^{(0)}$, then it fails on $Q^{(0)}$, so its average success probability is at most $1/2$.
Thus, in every case, the average success probability is at most $\frac{1}{2}+\frac{T}{2N}$.

The same bound holds after averaging over the internal randomness.
Since the algorithm succeeds with probability at least $1-\delta$ for any pair of $P$ and $Q$, we have that $ 1-\delta\leq\frac12+\frac{T}{2N}$.
Thus,
\begin{align*}
  T& \geq(1-2\delta)N
  =(1-2\delta)n\left\lfloor\frac q2\right\rfloor. \qedhere
\end{align*}
\end{proof}

\section*{Acknowledgement}

We thank Weiming Feng for some helpful comments.
Part of this project has received funding from the European Research Council (ERC) under the European Union's Horizon 2020 research and innovation programme (grant agreement No.~947778).

\bibliographystyle{alpha}
\bibliography{refs}

\newcommand{\etalchar}[1]{$^{#1}$}
\begin{thebibliography}{BGM{\etalchar{+}}25b}

\bibitem[BFS25]{BFS25}
Arnab Bhattacharyya, Weiming Feng, and Piyush Srivastava.
\newblock Approximating the total variation distance between {G}aussians.
\newblock In {\em AISTATS}, pages 1846--1854, 2025.

\bibitem[BGM{\etalchar{+}}24]{BGMMPV24}
Arnab Bhattacharyya, Sutanu Gayen, Kuldeep~S. Meel, Dimitrios Myrisiotis,
  A.~Pavan, and N.~V. Vinodchandran.
\newblock Total variation distance meets probabilistic inference.
\newblock In {\em ICML}, pages 3776--3794, 2024.

\bibitem[BGM{\etalchar{+}}25a]{BGMMPV25b}
Arnab Bhattacharyya, Sutanu Gayen, Kuldeep~S. Meel, Dimitrios Myrisiotis,
  A.~Pavan, and N.~V. Vinodchandran.
\newblock Total variation distance for product distributions is \#{P}-complete.
\newblock {\em Inf. Process. Lett.}, 189:106560, 2025.

\bibitem[BGM{\etalchar{+}}25b]{BGMMPV25a}
Arnab Bhattacharyya, Sutanu Gayen, Kuldeep~S. Meel, Dimitrios Myrisiotis, Aduri
  Pavan, and N.~V. Vinodchandran.
\newblock Computational explorations of total variation distance.
\newblock In {\em ICLR}, 2025.

\bibitem[FF26]{FF26}
Weiming Feng and Yucheng Fu.
\newblock On approximating the $f$-divergence between two {I}sing models.
\newblock In {\em {ITCS}}, volume 362 of {\em LIPIcs}, pages 59:1--59:23.
  Schloss Dagstuhl - Leibniz-Zentrum f{\"{u}}r Informatik, 2026.

\bibitem[FFYZ26]{FFYZ26}
Weiming Feng, Yucheng Fu, Minji Yang, and Anqi Zhang.
\newblock On computing total variation distance between mixtures of product
  distributions.
\newblock {\em arXiv}, abs/2605.03839, 2026.

\bibitem[FGJW23]{FGJW23}
Weiming Feng, Heng Guo, Mark Jerrum, and Jiaheng Wang.
\newblock A simple polynomial-time approximation algorithm for the total
  variation distance between two product distributions.
\newblock {\em TheoretiCS}, 2:Paper No. 7, 2023.

\bibitem[FLL24]{FLL24}
Weiming Feng, Liqiang Liu, and Tianren Liu.
\newblock On deterministically approximating total variation distance.
\newblock In {\em SODA}, pages 1766--1791, 2024.

\bibitem[FLY25]{FLY25}
Weiming Feng, Hongyang Liu, and Minji Yang.
\newblock Approximating the total variation distance between spin systems.
\newblock In {\em COLT}, pages 1974--2025, 2025.

\bibitem[Gla93]{Gla93}
Paul Glasserman.
\newblock Filtered {M}onte {C}arlo.
\newblock {\em Math. Oper. Res.}, 18(3):610--634, 1993.

\bibitem[Kon25]{Kon25}
Aryeh Kontorovich.
\newblock On the tensorization of the variational distance.
\newblock {\em Electron. Commun. Probab.}, 30:Paper No. 32, 10, 2025.

\bibitem[PTXZ26]{PTXZ26}
Eric Price, Kevin Tian, Zhiyang Xun, and Yusong Zhu.
\newblock Total variation distance estimation in autoregressive models.
\newblock {\em arXiv}, abs/2607.19510, 2026.

\end{thebibliography}

\end{document}